\pdfoutput=1
\documentclass[10pt,twocolumn]{article}
\usepackage[margin=0.85in]{geometry}
\usepackage{lmodern}
\usepackage[T1]{fontenc}
\usepackage{amsmath,amssymb,amsthm}
\usepackage{graphicx}
\usepackage{booktabs}
\usepackage{array}
\usepackage{microtype}
\usepackage[hidelinks]{hyperref}
\usepackage{xcolor}
\usepackage{float}
\usepackage{caption}
\newcounter{algctr}
\newenvironment{algo}[1]
 {\par\medskip\noindent\refstepcounter{algctr}%
  \begin{minipage}{\linewidth}\hrule\smallskip
  {\small\textbf{Algorithm \thealgctr.} #1}\smallskip\hrule\smallskip
  \small\ttfamily\obeylines\parindent=0pt\leftskip=1em}
 {\par\smallskip\hrule\medskip\end{minipage}\par\medskip}
\newcommand{\algcap}[1]{} 

\newtheorem{definition}{Definition}
\newtheorem{theorem}{Theorem}
\newtheorem{assumption}{Assumption}

\newcommand{\eco}{\textsc{eco}}
\newcommand{\cpo}{\textsc{cpo}}
\newcommand{\Commit}{\mathsf{Commit}}
\newcommand{\Sign}{\mathsf{Sign}}
\newcommand{\Hash}{\mathsf{H}}

\title{\textbf{ECO/CPO-DAG: A Contradiction-Based Accountability Layer\\for Adversarial Supply Chains}\\[2pt]
\large A Domain-Specific Cryptoeconomic Audit Protocol}

\author{Sebastian Cochinescu\\
\normalsize University of Bucharest\\
\normalsize \texttt{sebastian.cochinescu@drd.unibuc.ro}}
\date{}

\begin{document}
\maketitle

\begin{abstract}
\noindent
We present ECO/CPO-DAG, a domain-specific accountability protocol for
adversarial supply chains that formalizes \emph{contradiction detection}
as a supplemental validation layer rather than a consensus or
truth-establishing mechanism. Participants publish signed
\emph{Event Claim Objects} (ECOs) into a causally ordered, append-only
directed acyclic graph (DAG) whose edges encode happened-before
relations. When two claims about the same subject violate a
domain constraint, any observer can compile a \emph{Contradiction Proof
Object} (CPO)---a self-verifying object binding the two signed claims and
the violated rule---which, on public verification, triggers economic
slashing of a determinately blamed party. We map constraints to GS1
EPCIS~2.0 event semantics (spatial uniqueness, temporal monotonicity,
quantity conservation, quality monotonicity, regulatory validity),
so detection targets inconsistencies that are meaningful in practice.
Selective disclosure via commitment schemes and, optionally,
zero-knowledge contradiction proofs lets parties withhold claim
contents until a challenge forces the minimal opening.
We give an analytical treatment: an independent-observer detection
model $1-(1-p_{\min})^{h}$ (temporally, $1-e^{-h\lambda t}$), a
deterrence condition $S > g(1-p)/(kp)$ under $k$-party collusion, and a
storage estimate of order $1$~GB per participant per year under stated
assumptions. We are explicit about the protocol's boundary:
it detects \emph{provable} contradictions, not consistent lies; a party
that never contradicts itself is invisible to it, so the layer complements
---and does not replace---source verification and oracle aggregation.
A single-machine reference implementation corroborates the detection
model---the predicted coverage band overlaps the measured 95\% confidence
interval at every observer count---and records zero false accusations across
the runs; the fully zero-knowledge CPO, multi-party propagation, and
adaptive-adversary evasion remain analytical.
\end{abstract}

\section{Introduction}

Supply chains are the canonical setting in which parties who do not trust
one another must nonetheless rely on one another's records. A pharmaceutical
distributor must trust a carrier's cold-chain log; an automotive OEM must
trust a supplier's provenance attestation. The economic incentive to
misreport is real and well documented, and blockchain-based traceability
has been proposed repeatedly as a remedy~\cite{saberi2019blockchain}. Yet
recording a claim on a tamper-evident ledger~\cite{nakamoto2008bitcoin} does not make the claim
\emph{true}: the ``oracle problem''---that on-chain data is only as
trustworthy as the off-chain process that produced it---remains
unsolved~\cite{zhang2016towncrier,zhang2020deco}. A ledger can guarantee
that a party \emph{said} something and cannot later deny it; it cannot
guarantee that what the party said corresponds to physical reality.

This paper takes a deliberately narrow position. We do not attempt to
verify claims against ground truth. Instead we ask a weaker but tractable
question:

\begin{quote}\itshape
How can we economically disincentivize \emph{provable contradictions} in
semi-trusted supply chains through a domain-specific accountability layer?
\end{quote}

A provable contradiction is a pair of signed claims by identifiable parties
that cannot both hold under a domain rule---for example, a lot asserted to
originate in two mutually exclusive regions, or a sealed quantity that grows
between two custody events. Such contradictions are \emph{self-evident}: a
verifier needs only the two signed claims and the rule, not access to the
physical good. This is the accountability primitive that PeerReview
established for distributed systems~\cite{haeberlen2007peerreview} and that
the fault-detection literature formalized as detecting deviations from a
reference specification~\cite{haeberlen2009faultdetection}. We adapt it to
the supply-chain domain and couple it to a cryptoeconomic penalty.

\paragraph{Positioning.}
ECO/CPO-DAG is an \emph{audit and accountability} layer, not a consensus
protocol and not an oracle. It is complementary to
oracle aggregation (e.g.\ DECO's authenticated TLS
feeds~\cite{zhang2020deco} and adaptive conformal consensus over multiple
oracles~\cite{park2023acon2}), which govern how off-chain data
\emph{enters} the system; our layer governs \emph{post-hoc accountability}
once claims are on the record. It borrows the DAG data structure from
recent high-throughput consensus designs
(Narwhal/Tusk~\cite{danezis2022narwhal},
Bullshark~\cite{spiegelman2022bullshark},
DAG-Rider~\cite{keidar2021dagrider}) but uses the DAG only for causal
\emph{ordering and detectability}, not for agreement: each participant
maintains a local view, and no total order is required. Constraints are
aligned with GS1 EPCIS~2.0 event semantics~\cite{gs1epcis2022} so that
detection is defined over events practitioners already
emit.

\paragraph{Contributions.}
\begin{enumerate}\setlength{\itemsep}{1pt}\setlength{\topsep}{2pt}
\item A formal model of contradiction-based accountability for supply chains:
signed Event Claim Objects (ECOs) in a causally ordered append-only DAG,
and Contradiction Proof Objects (CPOs) with stated soundness properties
(\S\ref{sec:model}, \S\ref{sec:protocol}).
\item A mapping of five domain-constraint classes onto EPCIS event types,
making the ``domain-specific'' claim concrete (\S\ref{sec:protocol}).
\item A cryptoeconomic mechanism whose deterrence condition we derive
under rational and $k$-party-collusion adversaries (\S\ref{sec:economics}).
\item A privacy layer using commitment schemes and optional zero-knowledge
contradiction proofs, with an explicit disclosure taxonomy and an
acknowledgment that some contradiction classes cannot be checked privately
(\S\ref{sec:privacy}).
\item An analytical security and feasibility analysis---detection
probability, no-false-accusation, storage---with every quantitative claim
presented as a parameterized model output, together with a single-machine
reference implementation that corroborates the detection model against
measurement and reports measured costs and blame outcomes
(\S\ref{sec:security}--\S\ref{sec:eval}).
\end{enumerate}

\paragraph{Scope and honesty.}
The paper's core is analytical: all model numbers are outputs of the stated
models under stated parameters, marked as such throughout. A single-machine
reference implementation (\S\ref{sec:measured}) then \emph{corroborates} the
detection model against measurement and reports measured CPO costs, blame
outcomes, and storage. We are explicit about what it does not cover---the
fully zero-knowledge CPO, a multi-party network, and adaptive-adversary
evasion remain analytical---and its commitment/ZK figures are a conservative
upper bound (a MODP integer group and a naive Sigma protocol stand in for the
paper's curve-based Bulletproofs). The central limitation is intrinsic: the layer detects
\emph{contradictions}, so a party that lies consistently---never emitting a
claim that conflicts with its own prior claims---produces no CPO and is not
penalized. We regard this not as a flaw to be hidden but as the precise
boundary of what contradiction-based accountability can offer. A second
boundary is narrower still: on-chain \emph{slashing} fires only where blame is
attributable from the contradicting pair alone---self-equivocation or a
single-issuer monotone violation (\S\ref{sec:economics}); every other
contradiction yields a CPO that stands as portable evidence for off-chain
adjudication, not an automatic penalty.

\section{Related Work}
\label{sec:related}

\paragraph{Accountability in distributed systems.}
PeerReview \cite{haeberlen2007peerreview} makes nodes accountable by having
each keep a tamper-evident log and having others check it against a
reference implementation, exposing any observable deviation. The
fault-detection problem~\cite{haeberlen2009faultdetection} formalizes which
faults are detectable at all. Non-equivocation---preventing a party from
telling different things to different peers---has been enforced with trusted
hardware in A2M~\cite{chun2007a2m} and TrInc~\cite{levin2009trinc}, and
transparency logs (Certificate Transparency~\cite{laurie2013ct},
CONIKS~\cite{melara2015coniks}) make equivocation over a public log
detectable without trusted hardware. Accountable Universal Composability
(AUC)~\cite{graf2023auc} gives a
composable definition of accountability. Our CPO is in this lineage: it is
a portable, self-verifying witness of a specific detectable fault
(a domain-rule contradiction), but we add a \emph{cryptoeconomic}
consequence rather than only exposure.

\paragraph{DAG-structured logs and consensus.}
Causal DAGs underlie modern high-throughput BFT
(Narwhal/Tusk~\cite{danezis2022narwhal},
Bullshark~\cite{spiegelman2022bullshark},
DAG-Rider~\cite{keidar2021dagrider}), which raise the throughput of classical
leader-based BFT~\cite{castro2002pbft}, and non-blockchain ledgers. We reuse
the DAG for happened-before ordering via hybrid logical
clocks~\cite{lamport1978clocks,kulkarni2014hlc} but explicitly do
\emph{not} run agreement over it; the design goal is auditability, not a
total order.

\paragraph{Oracles.}
Town Crier~\cite{zhang2016towncrier} and DECO~\cite{zhang2020deco}
authenticate data as it enters a smart contract;
Ekiden~\cite{cheng2019ekiden} adds confidential execution;
ACon$^2$~\cite{park2023acon2} aggregates multiple oracle reports with
statistical guarantees. These address \emph{ingestion}. Our layer is
downstream and orthogonal: it operates on claims already recorded and asks
whether they are mutually consistent.

\paragraph{Cryptoeconomics.}
Slashing of staked deposits for provable misbehavior is standard in
proof-of-stake finality~\cite{buterin2017casper}; the broader study of
rational adversaries and extractable value~\cite{daian2020flashboys}
motivates a game-theoretic treatment. We apply the slashing idea to
provable domain contradictions rather than protocol-level equivocation.

\paragraph{Fraud proofs.}
Structurally, a CPO is a \emph{domain-specific fraud proof}: a portable,
self-verifying witness of an invalid transition that triggers an economic
penalty. This is the primitive behind optimistic-rollup fault proofs and the
fraud/data-availability proofs of Al-Bassam, Sonnino, and
Buterin~\cite{albassam2018fraud}, where a single honest party can prove an
invalid state transition to everyone else. We specialize that idea to GS1
EPCIS domain rules with \emph{determinate blame}: the fraud proved is a
supply-chain-constraint violation by an identifiable issuer, not an invalid
rollup state root.

\paragraph{Cryptographic building blocks and supply chains.}
We use Pedersen commitments~\cite{pedersen1991vss}, Merkle
authentication~\cite{merkle1987signature}, succinct
proofs~\cite{groth2016size}, and CL-style (Camenisch--Lysyanskaya) selective-disclosure
credentials~\cite{camenisch2004signatures}. On the domain side, GS1
EPCIS~2.0~\cite{gs1epcis2022} supplies the event vocabulary; prior
blockchain-for-supply-chain work~\cite{saberi2019blockchain} motivates the
setting. Relative to that body of work, our differentiator is native
\emph{contradiction detection with determinate blame and slashing}, defined
over EPCIS events, rather than mere immutable recording.

Table~\ref{tab:compare} positions ECO/CPO-DAG against the closest mechanisms:
prior accountability logs detect equivocation but attach no economic
consequence and carry no domain semantics; staking systems slash but do not
detect cross-event domain contradictions. Ours is the only one combining
domain-rule contradiction detection, determinate blame, slashing, and
selective-disclosure privacy.

\begin{table}[t]
\centering
\small
\setlength{\tabcolsep}{3pt}
\begin{tabular}{@{}p{2.15cm}ccccc@{}}
\toprule
\textbf{Mechanism} & \textbf{Detect} & \textbf{Blame} & \textbf{Slash} & \textbf{EPCIS} & \textbf{Priv.} \\
\midrule
PeerReview \cite{haeberlen2007peerreview} & yes & part.\ & no & no & no \\
CT/CONIKS \cite{laurie2013ct,melara2015coniks} & yes & yes & no & no & no \\
A2M/TrInc~\cite{chun2007a2m,levin2009trinc} & yes & yes & no & no & no \\
PoS slashing~\cite{buterin2017casper} & part.\ & yes & yes & no & no \\
Oracle staking & no & yes & yes & no & no \\
\textbf{This work} & \textbf{yes} & \textbf{yes} & \textbf{yes} & \textbf{yes} & \textbf{yes} \\
\bottomrule
\end{tabular}
\caption{Positioning against related mechanisms. \textbf{Detect}=contradiction
detection (PeerReview: spec deviation; CT/CONIKS, A2M/TrInc: equivocation; PoS:
protocol-level equivocation only; oracle staking: no \emph{native cross-event
domain-rule} detection, as it aggregates independent reports; ours: domain-rule
contradictions).
\textbf{Blame}=determinate blame (ours bounded to self-equivocation and
monotone violations; other contradictions yield $\bot$). \textbf{Slash}=economic
slashing. \textbf{EPCIS}=domain-specific supply-chain semantics.
\textbf{Priv.}=selective-disclosure privacy.}
\label{tab:compare}
\end{table}

\section{System Model}
\label{sec:model}

\paragraph{Notation.}
Symbols used in the analysis: $\kappa$ (constraint class); $h$ (number of
honest observers/watchtowers); $p, p_{\min}$ (aggregate and per-observer
detection probability); $\lambda$ (per-observer detection rate); $t$ (detection
horizon); $S$ (stake); $\alpha$ (overcollateralization factor); $V$ (maximum
single-transaction value); $R$ (history-based risk factor); $g$ (one-shot gain
from lying); $k$ (collusion size); $D$ (CPO anti-spam deposit); and $f$
(tolerated Byzantine faults among the $n$ checkpoint participants).

\subsection{Objects}

\begin{definition}[Event Claim Object]
An \eco{} is a tuple
\[
\eco = \langle \mathrm{id},\, \mathrm{cm},\, \sigma,\, \tau,\, \mathrm{refs},\, \mathrm{subj} \rangle
\]
where $\mathrm{cm} = \Commit(\mathrm{claim}, r)$ is a hiding, binding
commitment to the claim payload under randomness $r$;
$\mathrm{subj}$ names the subject of the claim (e.g.\ a GS1 lot or item
identifier); $\tau$ is a hybrid logical clock timestamp;
$\mathrm{refs}$ is the set of causal-parent \eco{} identifiers;
$\mathrm{pk}$ is the issuer's public key;
$\mathrm{id} = \Hash(\mathtt{dsep}\,\|\,\mathrm{pk}\,\|\,\mathrm{cm}\,\|\,\tau\,\|\,\mathrm{refs}\,\|\,\mathrm{subj})$
with a domain-separation tag $\mathtt{dsep}$, so two issuers committing to
identical fields still mint distinct ids; and
$\sigma = \Sign(\mathit{sk}, \mathrm{id})$ binds the object to the issuer's
identity.
\end{definition}

The commitment $\mathrm{cm}$ lets an \eco{} be published without revealing
its payload; the payload is opened only when a challenge requires it
(\S\ref{sec:privacy}).

\begin{definition}[Contradiction Proof Object]
A \cpo{} is a tuple
\[
\cpo = \langle \eco_1,\, \eco_2,\, \kappa,\, \pi,\, d,\, D \rangle
\]
asserting that $\eco_1$ and $\eco_2$ jointly violate constraint
$\kappa \in \{\textsf{spatial}, \textsf{temporal}, \textsf{quantity},
\textsf{quality}, \textsf{regulatory}\}$. The witness $\pi$ is either the
two openings of $\mathrm{cm}_1, \mathrm{cm}_2$ together with a check that
$\kappa$ is violated, or a zero-knowledge proof of the same statement
(\S\ref{sec:privacy}). The challenger identity is $d$, and $D$ is a
refundable anti-spam deposit.
\end{definition}

\begin{definition}[Supply-Chain DAG]
The log is a directed acyclic graph $G = (V, E)$ with $V$ the set of
published \eco{}s and $E$ the causal edges induced by $\mathrm{refs}$.
The clock invariant is
$\forall (u \to v) \in E:\ \tau(u) < \tau(v)$
under the hybrid-logical-clock order~\cite{kulkarni2014hlc}, giving a
partial happened-before order~\cite{lamport1978clocks}. The graph is
append-only; there is no global agreement on $V$, and each participant holds
a local view $G_i \subseteq G$.
\end{definition}

Figure~\ref{fig:dag} illustrates the structure and a same-issuer
equivocation witnessed by a CPO.

\begin{figure}[t]
\centering
\includegraphics[width=\linewidth]{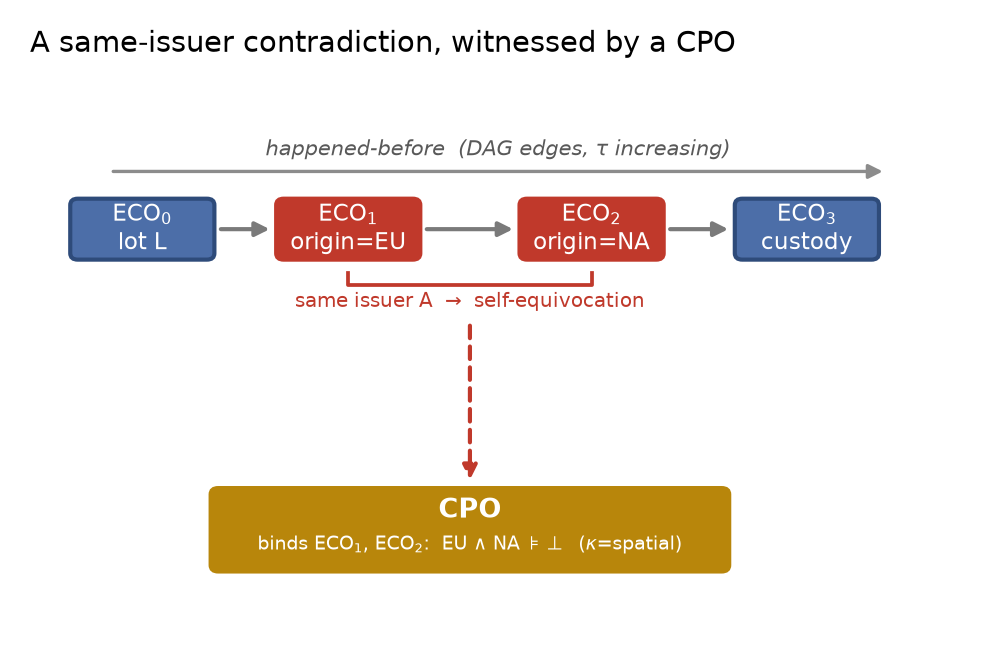}
\caption{Signed claims (ECOs) form a causally ordered, append-only DAG whose
edges are happened-before relations. A single issuer's two claims about one
subject (lot~$L$) assert incompatible origins; because both carry that
issuer's valid signatures, any observer can compile a CPO---a self-verifying
witness that
$\mathrm{origin}(L){=}\text{EU} \wedge \mathrm{origin}(L){=}\text{NA}
\models \bot$. The CPO proves the contradiction unconditionally; it
pins blame only when \S\ref{sec:economics} resolves determinate fault---here,
self-equivocation by one issuer. A contradiction between two different
issuers would be equally provable but need not be attributable, and may
resolve to $\bot$.}
\label{fig:dag}
\end{figure}

\subsection{Assumptions}

We separate the assumption needed to \emph{detect} a contradiction from the
stronger one needed to \emph{disseminate} the proof reliably.

\begin{assumption}[Detection]
\label{asm:detection}
At least one honest observer exists with (i) read access to the relevant
same-subject \eco{}s (every claim sharing the subject --- not only the
subject's declared causal ancestors) and (ii) the ability to broadcast a
\cpo{}.
Detection requires only this single honest observer; it does not require an
honest majority. This is a \emph{possibility} statement: a bounded-work
observer samples only a fraction of candidate pairs, so its per-observer
detection probability is $p_{\min}<1$ rather than certain; \S\ref{sec:security}
models this and shows how coverage grows with the number of observers $h$.
\end{assumption}

\begin{assumption}[Propagation]
The network is partially synchronous with an eventual but unknown message
delay bound~\cite{dwork1988partialsynchrony}, and a Byzantine reliable
broadcast primitive~\cite{bracha1987brb} is available. For the checkpoint
and finality claims of \S\ref{sec:security} we assume at most $f < n/3$
Byzantine participants among the $n$ that maintain checkpoints.
\end{assumption}

\begin{assumption}[Cryptography]
The hash $\Hash$ is collision-resistant, the signature scheme is
existentially unforgeable, and the commitment $\Commit$ is computationally
binding and hiding~\cite{pedersen1991vss}.
\end{assumption}

\subsection{Adversary}

Identities are permissioned (each participant holds a certified key) while
\emph{actions} are permissionless (any identity may emit any claim). We
consider a hybrid adversary: \emph{rational} parties maximize expected
profit and respond to incentives, and may additionally act
\emph{Byzantine}---emitting arbitrary or griefing messages---at a cost. The
adversary may equivocate, backdate within clock tolerance, collude in
groups of size $k$, and attempt to spam CPOs. The adversary cannot forge
signatures or hash collisions, and cannot suppress messages indefinitely
after the global stabilization time under Assumption~2.

\section{Protocol Design}
\label{sec:protocol}

\subsection{Emitting a claim}

A party creates an \eco{} by committing to its claim, timestamping it with
the local hybrid logical clock, linking it to its causal parents, and
signing the resulting identifier (Algorithm~\ref{alg:eco}). Publishing the
commitment rather than the plaintext is what enables selective disclosure:
the object is on the record and non-repudiable, but its contents remain
hidden until opened.

\begin{algo}{\textnormal{\textsc{CreateECO}}$(\mathrm{claim},\mathrm{subj},\mathrm{refs})$}
\label{alg:eco}
1: $r \gets$ fresh randomness
2: $\mathrm{cm} \gets \Commit(\mathrm{claim}, r)$
3: $\tau \gets \textsc{HLC.now}()$
4: $\mathrm{id} \gets \Hash(\mathtt{dsep}\|\mathrm{pk}\|\mathrm{cm}\|\tau\|\mathrm{refs}\|\mathrm{subj})$
5: $\sigma \gets \Sign(\mathit{sk}, \mathrm{id})$
6: $\eco \gets \langle \mathrm{id}, \mathrm{cm}, \sigma, \tau, \mathrm{refs}, \mathrm{subj}\rangle$
7: \textsc{Broadcast}($\eco$); store $r$ locally; \textbf{return} $\eco$
\end{algo}

\subsection{Maintaining the DAG}

Each participant appends received \eco{}s to its local view after checking
the signature and that all causal parents in $\mathrm{refs}$ are present and
clock-consistent. There is no ordering vote: the DAG is a shared
append-only substrate, and disagreement about membership is tolerated
because detection needs only that \emph{one} honest party sees both sides of
a contradiction. Claim emission and contradiction detection are therefore
\emph{consensus-free}: they require no agreement on $V$. Only the
\emph{optional} checkpoint pruning and CPO-finality claims of
\S\ref{sec:security} invoke Byzantine reliable broadcast and the $f<n/3$
bound (Assumption~2); a deployment that forgoes pruning and on-chain finality
runs the accountability layer with no BFT layer at all. Periodic checkpoints
(\S\ref{sec:security}) allow old history to be pruned to a Merkle
root~\cite{merkle1987signature}.

\subsection{Detecting a contradiction}

On receiving a new \eco{}, an observer indexes claims by subject and tests
the new claim against same-subject claims in its local view. To bound work
and resist denial of service, $\mathrm{refs}$ and radius limits prioritize
the same-subject walk under load, and each sender is rate-limited; a
candidate CPO carries a refundable deposit (Algorithm~\ref{alg:detect}).

\begin{algo}{\textnormal{\textsc{Detect}}$(\eco_{\mathrm{new}}, G_i)$ \textnormal{(rate-limited)}}
\label{alg:detect}
1: \textbf{if} \textsc{RateExceeded}($\eco_{\mathrm{new}}.\mathrm{issuer}$) \textbf{then return}
2: $A \gets \textsc{SubjectClaims}(\eco_{\mathrm{new}}, G_i)$
3: \textbf{for} $\eco_{\mathrm{old}} \in A$, $\kappa \in \textsf{Constraints}$ \textbf{do}
4: \quad \textbf{if} \textsc{Violates}($\eco_{\mathrm{new}}, \eco_{\mathrm{old}}, \kappa$) \textbf{then}
5: \quad\quad $D \gets \textsc{LockDeposit}()$
6: \quad\quad $\cpo \gets \textsc{MakeCPO}(\eco_{\mathrm{new}}, \eco_{\mathrm{old}}, \kappa, D)$; \textsc{Broadcast}($\cpo$)
7: $G_i.\textsc{append}(\eco_{\mathrm{new}})$
\end{algo}

\paragraph{Detection scope (omitted parents).}
\textsc{SubjectClaims} is not limited to the parents an ECO declares in
$\mathrm{refs}$: an observer indexes every ECO in its local view by
$\mathrm{subj}$ and checks a new claim against same-subject claims it holds,
not only the declared causal ancestors. Omitting a known parent to fork a
subject into a disconnected timeline therefore does not evade detection; two
same-subject claims are comparable whether or not $\mathrm{refs}$ links them,
and a missing-but-expected link is itself evidence that can seed a temporal
or spatial CPO. The $\mathrm{refs}$ field and the radius limit are an
\emph{optimization} that orders and bounds the walk under load; they are not the
detection boundary. The bound is not free, however: under sustained load an
adversary can aim a contradiction at a causally distant same-subject pair the
budgeted radius deprioritizes, so completeness over same-subject pairs holds
only up to the per-observer work budget --- that residual is exactly what the
sampled per-observer coverage $p_{\min}$ of \S\ref{sec:security} models, not a
guarantee of exhaustive checking.

\paragraph{What observers can inspect.}
Detection needs field values, yet payloads sit behind commitments, so
\textsc{Violates} runs over whatever disclosure tier a subject exposes:
(a) fields published in plaintext, such as many EPCIS routing fields (location, timestamp, subject id);
(b) fields the issuer selectively discloses under a standing policy or on
challenge (\S\ref{sec:privacy}); (c) fields a \emph{watchtower} --- an honest
observer that monitors the DAG for contradictions, in the sense of a
Lightning-Network watchtower~\cite{mccorry2019pisa} --- is explicitly authorized to decrypt; or
(d) predicates checkable in zero knowledge over
the commitments (\S\ref{sec:privacy}) without opening them. A contradiction
is discoverable exactly when the fields its rule touches fall in one of
these tiers. A subject that commits everything and authorizes nothing is
auditable only via the ZK/disclosure paths discussed in \S\ref{sec:privacy}
(summarized in Table~\ref{tab:disclosure}), and some classes (such as quality
over a chain) then cannot be checked at all without disclosure. This is the
concrete privacy/detection tradeoff the paper flags in \S\ref{sec:privacy}.

This resolves an apparent tension with the Detectability theorem
(\S\ref{sec:security}), which wants \emph{many} independent watchtowers (large
$h$) while the privacy layer wants \emph{few} holders of the openings. The two
coexist because detection scales with $h$ only over tiers (a), (b), and (d) ---
fields that are public, disclosed-by-policy, or ZK-checkable without opening.
Where a constraint's fields stay fully committed and undisclosed, detection is
confined to the parties that already hold the openings --- typically the two
counterparties to the event --- and a ZK-CPO lets such a holder prove the
contradiction to everyone else without revealing it. Growing $h$ therefore
strengthens detection for the public and ZK-checkable constraints; it does not
conjure openings that no watchtower holds. A deployment picks, per field, how
far up the detectability-vs-confidentiality curve it wants to sit.

\subsection{Domain constraints over EPCIS}
\label{sec:constraints}

The constraint set is where the protocol becomes domain-specific.
Table~\ref{tab:epcis} maps each constraint class to the EPCIS~2.0 event type
it is checked against and the rule that defines a violation. These are the
predicates \textsc{Violates} evaluates. Each rule is a \emph{pairwise}
predicate over two same-subject events, keeping every class inside the two-ECO
CPO form for \emph{detection} (the quality class carries one further
qualification for \emph{slashing}, noted at the end of this section); the
regulatory class in particular pairs a compliance assertion with
a revocation or expiry record for the cited certificate, rather than flagging a
lone missing field (a single-event validity check, which sits outside the
contradiction primitive). For that pair to stay inside the DAG's trust base,
the certificate lifecycle must itself be represented as ECOs: the revocation
signed by the issuing authority acting as a participant, and expiry carried as
a validity interval inside the signed certificate claim rather than read from a
wall clock --- otherwise the regulatory class alone would import an external
truth oracle. The temporal rule warrants a note: the append-time
clock check (\S\ref{sec:protocol}) and the DAG invariant (\S\ref{sec:model})
already force monotone timestamps along every \emph{declared} edge, so a
temporal CPO is meaningful only among same-subject events that are
\emph{$\mathrm{refs}$-incomparable} (no path between them) --- it flags a pair
whose asserted event-times cannot be reconciled with the happened-before order
(backdating), and $\mathrm{refs}$-incomparable \emph{concurrent} events are
never a violation. Backdating \emph{within} the clock tolerance granted in
\S\ref{sec:model} is reconcilable by construction, so the temporal class is
blind up to that tolerance --- the tolerance is a deployment parameter trading
detection sharpness against clock-synchronization strictness. Because temporal contradictions across issuers resolve to
$\bot$, a spurious one is a false-CPO (spam) risk bounded by the deposit $D$,
not a false-slash risk. The quantity rule is read the same pairwise way: it
pairs an issuer's input-quantity claim against its own later output-quantity
claim for the subject (the EPCIS TransformationEvent input/output quantity
lists), so a mass-balance violation $\sum\text{outputs} > \sum\text{inputs}$ is
attributable to that single issuer from the pair. It is deliberately \emph{not}
read off a single AggregationEvent, which in EPCIS~2.0 models reversible
containment (packing cases onto a pallet, then disaggregating) where child
counts are conserved by definition --- there is no input-vs-output balance to
violate there.

The quality rule needs one qualification the closed classes do not.
\emph{Detecting} a candidate quality contradiction is pairwise --- two of the
issuer's own condition claims, one improving on the other --- but whether the
improvement is \emph{illegitimate} depends on the absence of an intervening
authorized transformation or certification, an existential condition over the
chain rather than a function of the two ECOs alone (a partial-view observer
missing the authorizing event would see a violation that isn't one). So the
quality class \emph{slashes} only in its closed form --- a
\textbf{self-equivocation on condition}, where the two same-issuer claims are
jointly impossible regardless of any intervening event (e.g.\ a
continuous-cold-chain assertion contradicted by that same issuer's own excursion
log). A monotone quality \emph{improvement} that an unobserved authorized
transformation could explain is not attributable from the pair alone and
resolves to $\bot$, exactly like the cross-issuer cases --- never an automatic
slash. This is what keeps the No-false-accusation guarantee
(\S\ref{sec:security}) intact for the quality class; unlike quantity, whose
$\sum\text{outputs} > \sum\text{inputs}$ check closes over the two ECOs, quality
closes over the pair only in the self-equivocation form.

\begin{table}[t]
\centering
\small
\begin{tabular}{@{}p{1.7cm}p{2.55cm}p{3.2cm}@{}}
\toprule
\textbf{Constraint} & \textbf{EPCIS event} & \textbf{Violation rule} \\
\midrule
Spatial & ObjectEvent & Two locations for one subject at one time \\
Temporal & all events & Backdating: two same-subject events whose asserted times contradict their happened-before order (concurrent events are not a violation) \\
Quantity & Transformation\-Event (input/output quantity lists) & $\sum\text{outputs} > \sum\text{inputs}$ across an issuer's own input- and output-quantity claims for the subject \\
Quality & ObjectEvent (disposition / sensor\-Element\-List) & Two same-issuer condition claims for the subject that are jointly impossible (self-equivocation on condition). A monotone \emph{improvement without authorization} is detectable but resolves to $\bot$ --- it is not closed over the pair (\S\ref{sec:constraints}) \\
Regulatory & Transaction\-Event & Compliance asserted against a certificate contradicted by a revocation/expiry record \\
\bottomrule
\end{tabular}
\caption{Domain constraints mapped to GS1 EPCIS~2.0 event
types~\cite{gs1epcis2022}. A CPO cites exactly one constraint class and the
two events whose pairing violates it.}
\label{tab:epcis}
\end{table}

\subsection{Challenging and adjudicating}

A published \cpo{} is verified by anyone: check both signatures, check that
the openings (or the zero-knowledge proof) are valid, and check that the
cited constraint is indeed violated. If verification succeeds and blame is
determinate (\S\ref{sec:economics}), the blamed party's stake is slashed and
the deposit-backed challenger is rewarded; if verification fails, the
challenger forfeits its deposit $D$. Figure~\ref{fig:pipeline} shows the
end-to-end flow.

\begin{figure}[t]
\centering
\includegraphics[width=\linewidth]{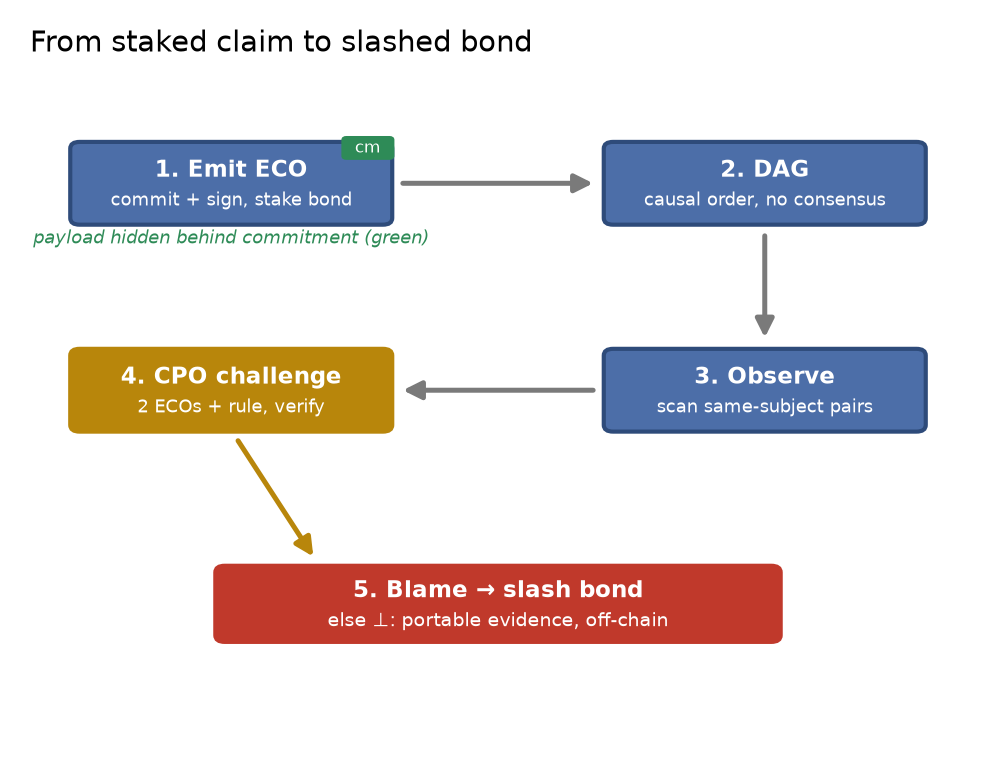}
\caption{The accountability layer. Parties emit ECOs against staked bonds;
the DAG orders them causally; observers monitor for constraint violations;
a CPO challenge is publicly verified and, on determinate blame, slashes the
bond. Commitments (green) let claim contents stay hidden until a challenge
forces the minimal opening.}
\label{fig:pipeline}
\end{figure}

\section{Cryptoeconomic Mechanism}
\label{sec:economics}

\subsection{Stake and blame}

Each participant $i$ posts a stake
$S_i = \alpha \, V_i \, R_i$,
where $V_i$ is the maximum single-transaction value it can assert,
$R_i \geq 1$ a history-based risk factor, and $\alpha > 1$ an
overcollateralization factor. Slashing requires \emph{determinate blame}:
a CPO is necessary evidence but not sufficient, because a bare contradiction
may not identify which of two parties is at fault.

\begin{definition}[Blame]
$\textsc{Blame}(\eco_1, \eco_2, \kappa) \to \{\eco_1.\text{issuer},\,
\eco_2.\text{issuer},\, \bot\}$ resolves fault by:
(i) \emph{self-equivocation}---if a single issuer signed both $\eco_1$ and
$\eco_2$ and the two claims are \emph{jointly unsatisfiable} (they cannot both
be true: two locations for one subject at one time, two irreconcilable
event-times, a jointly-impossible condition---e.g.\ a continuous-cold-chain
assertion against that issuer's own excursion log---or a compliance assertion
against that same issuer's own revocation/expiry record), blame that issuer;
(ii) \emph{conservation (monotone) violation}---if $\kappa$ is a conservation
constraint whose two claims are each individually satisfiable yet together
breach the domain rule (the quantity case: an output exceeding an input the
same issuer acknowledged) and the offender is attributable from the pair alone
(both events on that one issuer's self-authored chain), blame
that issuer; (iii) otherwise return $\bot$ (ambiguous), deferring to
off-chain adjudication. A contradiction between two different issuers'
single claims, where neither self-equivocates and no party is attributable
from the pair alone, is therefore always $\bot$: the protocol does not guess
which party lied. In particular, a quantity or quality mismatch that spans
two issuers---for example, a sender's under-reported input against a
receiver's truthful larger reading---is not automatically slashed; lacking
attribution from the pair alone it resolves to $\bot$, so an honest receiver
is never penalized for a counterparty's under-report. A \emph{same-issuer}
quality improvement is likewise $\bot$: it is not jointly unsatisfiable---an
authorized transformation could reconcile the two condition claims---and it is
not a conservation violation, so it satisfies neither (i) nor (ii)
(\S\ref{sec:constraints}). Slashing fires only when blame is not $\bot$.
\end{definition}

We deliberately reject a \emph{last-writer-wins} heuristic (blame whichever
party wrote later): it would slash an honest party whose truthful claim
merely contradicts a counterparty's earlier falsehood, violating the
No-false-accusation theorem (\S\ref{sec:security}). Restricting determinate
blame to the two cases in which fault is attributable from the pair alone is
what makes that theorem hold; all other contradictions are surfaced as
$\bot$-blame CPOs for off-chain adjudication rather than automatic slashing.

\paragraph{Which classes carry economic teeth.}
It follows that a slash fires only on (a) \emph{self-equivocation} --- one
issuer signing two \emph{jointly-unsatisfiable} claims that cannot both be
true; this covers spatial, temporal, same-issuer regulatory, and
\textbf{quality} (two of that issuer's own condition claims that cannot both
hold, e.g.\ a continuous-cold-chain assertion contradicted by that issuer's own
excursion log) --- or (b) a \emph{conservation violation closed over the pair}:
the \textbf{quantity} case, one issuer's own input claim against its own later
output claim, where $\sum\text{outputs} > \sum\text{inputs}$ is a self-contained
arithmetic check on two individually-consistent claims. Both are the same
two-ECO pairing used throughout (\S\ref{sec:protocol}), attributable to one
issuer from the pair alone, \emph{not} a single event checked in isolation. A
quality \emph{improvement} that an unobserved authorized transformation could
explain is \textbf{neither} jointly unsatisfiable \textbf{nor} a conservation
violation --- confirming that no authorization exists is a negative over a
possibly-partial view --- so it resolves to $\bot$ (off-chain), never an
automatic slash (\S\ref{sec:constraints}). Cross-issuer \textsf{spatial},
\textsf{temporal}, and \textsf{regulatory} contradictions, absent
self-equivocation, likewise resolve to $\bot$: surfaced as evidence for
off-chain adjudication, not slashed. The enforceable core is therefore
self-equivocation (including on condition) plus single-issuer quantity
conservation. We state this plainly because it is arguably a
larger practical limit than the consistent-liar blind spot we foreground
elsewhere: much of the value lies in \emph{creating portable, self-verifying
evidence} (the CPO) even where on-chain slashing does not fire.

\subsection{Deterrence condition}

Consider a party choosing between honest reporting and lying for a one-shot
gain $g$, where $g$ is the gain from lying \emph{net of} the honest payoff ---
that payoff is collected either way, so it cancels from the comparison and $g$
is the marginal gain. If detection probability is $p$ and the party is
slashed $S$ on detection, its expected utility from lying, relative to
reporting honestly, is
\begin{equation}
EU_{\text{lie}} = (1-p)\,g - p\,S .
\label{eq:eulie}
\end{equation}
Lying is deterred ($EU_{\text{lie}} < 0$) exactly when
\begin{equation}
S > \frac{(1-p)}{p}\, g .
\label{eq:deter}
\end{equation}
For $k$ parties who jointly falsify a single event and split one gain $g/k$
while each stakes $S$, the per-party deterrence condition becomes
\begin{equation}
S > \frac{(1-p)}{k\,p}\, g .
\label{eq:collude}
\end{equation}
Note the direction of Eq.~\eqref{eq:collude}: dividing a fixed gain across
more colluders makes each one easier to deter, so this is a best-case
reading that applies only to shared-gain joint falsification. It does not
model the more dangerous collusion in which parties cover for one another
and thereby lower the detection probability $p$ --- declining to report each
other, or coordinating claims so no checkable pair ever forms --- which raises
the stake required and is instead bounded by the detection assumption (a single
honest observer suffices, Assumption~\ref{asm:detection}) and by the watchtower
count $h$ in the Detectability theorem (\S\ref{sec:security}). Where each
colluder reaps a full gain $g$ rather than a shared
$g/k$, Eq.~\eqref{eq:deter} applies per party unchanged.

Eq.~\eqref{eq:eulie} carries one further assumption worth surfacing: it treats
$+g$ and $-S$ as mutually exclusive, i.e., detection \emph{reverses} the gain.
In a physical supply chain the fraud is often consummated before a CPO is
compiled --- the counterfeit part installed, the payment cleared --- so the
liar retains $g$ and only the stake is at risk:
$\mathrm{EU}_{\mathrm{lie}} = g - p\,S$, and deterrence needs $S > g/p$ (per
party, $S > g/(kp)$ in the shared-gain case). At $p = 0.5$ that is roughly
twice the stake of Eq.~\eqref{eq:deter}. A deployment whose lies are not
clawed back on detection should size $\alpha$ from the retained-gain
condition, which we treat as the conservative design target.

Equations~\eqref{eq:deter}--\eqref{eq:collude} are the design targets for
setting $\alpha$; Figure~\ref{fig:econ} plots them. We stress these are
\emph{model} relations under risk-neutral, one-shot assumptions; repeated
play and risk aversion only strengthen deterrence, while reputational value
of $g$ that exceeds a single transaction weakens it.

\begin{figure}[t]
\centering
\includegraphics[width=\linewidth]{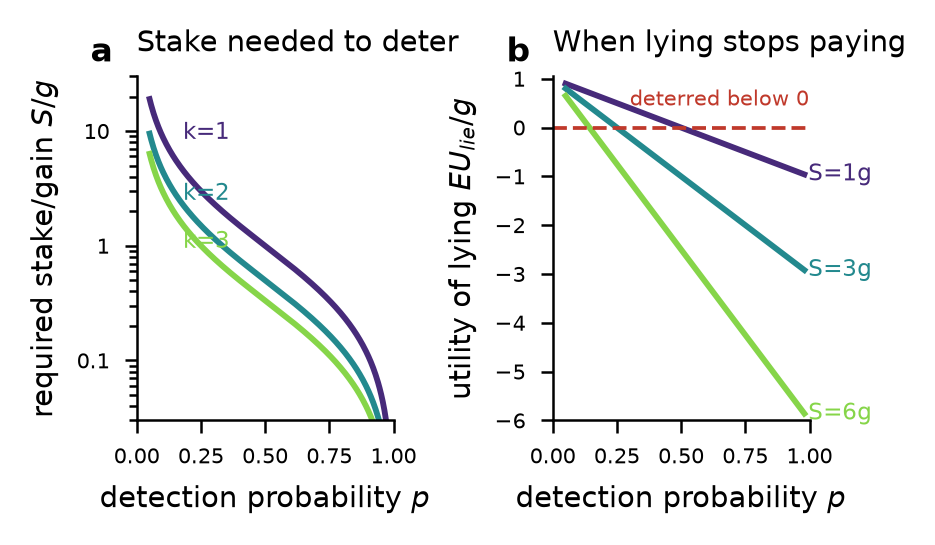}
\caption{Cryptoeconomic model outputs (not measurements). \textbf{Left:} the
stake-to-gain ratio required for deterrence, $S/g > (1-p)/(kp)$, versus
detection probability $p$ for collusion size $k$; higher detection lowers the
required stake, while larger $k$ lowers it only under the shared-single-gain
reading of Eq.~\eqref{eq:collude}. Collusion that instead suppresses
detection moves the requirement the other way.
\textbf{Right:} expected utility of lying from Eq.~\eqref{eq:eulie} for
several stake multiples $S{=}\{1,3,6\}g$; lying is deterred where the curve
crosses below zero (red line).}
\label{fig:econ}
\end{figure}

\subsection{Slashing and bounty}

On a valid CPO with determinate blame, the blamed stake is slashed and
redistributed. A concrete split---challenger/detector, corroborating
reporters, protocol treasury---trades off incentivizing detection against
discouraging frivolous or collusive self-challenges; the deposit $D$ on each
CPO bounds the cost of spam (\S\ref{sec:security}). We treat the exact split
as a governance parameter rather than a fixed constant, since its optimal
value depends on the deployment's watchtower economics. That dependence
deserves emphasis: monitoring is costly (\S\ref{sec:measured} measures
\textsc{Detect}() growing super-linearly in subject history) while a given
contradiction pays only its first reporter, so detection is a public good with
a free-rider structure --- rational watchtowers under-invest, and the $h$ that
the Detectability theorem takes as given is endogenous to this bounty design.
We flag it as an assumption of the same rank as the consistent-liar boundary:
the analysis quantifies coverage \emph{given} $h$; it does not derive $h$.

\section{Privacy Layer}
\label{sec:privacy}

\subsection{Commitment-based claims}

The public part of an \eco{} is $\langle \mathrm{id}, \tau, \mathrm{refs},
\mathrm{subj}, \mathrm{cm}\rangle$; the payload stays private under the
Pedersen commitment $\mathrm{cm}$~\cite{pedersen1991vss}, whose additive
homomorphism lets the aggregate of committed quantities be formed without
opening them. Because a quantity violation is an inequality
($\sum\text{outputs} > \sum\text{inputs}$), not an equality, checking it on
committed values requires a zero-knowledge range proof over the homomorphic
difference~\cite{bunz2018bulletproofs}, not the homomorphism alone: the
homomorphism yields a commitment to the difference, and the range proof
certifies its sign without revealing the amounts. An honest party can
therefore participate---and be held accountable---without publishing
commercially sensitive volumes, prices, or routes.

\subsection{Selective disclosure and ZK contradiction proofs}

When a challenge requires it, a party can open exactly the fields the
constraint touches, using CL-style (Camenisch--Lysyanskaya) selective-disclosure
credentials~\cite{camenisch2004signatures}: reveal a timestamp while hiding
quantity, or prove an aggregate without revealing components. A CPO can
itself be zero-knowledge: the challenger proves
\emph{``I know openings of $\mathrm{cm}_1, \mathrm{cm}_2$ under valid
signatures such that constraint $\kappa$ is violated''} without revealing
the claims, via a succinct proof~\cite{groth2016size}. For Groth16 we report only the
scheme's \emph{published} asymptotics (constant-size proof, constant-time
verification); we do not report gas costs. For the quantity range proof our
reference implementation (\S\ref{sec:measured}) reports measured size and
prove/verify time---but over a MODP integer group with a naive Sigma protocol,
so those figures are a conservative upper bound, not the curve-based
Bulletproofs a deployment would use. Groth16 also requires a per-circuit
trusted setup; a deployment unwilling to run one can substitute a
transparent-setup system (e.g.\ a Bulletproofs range proof for the quantity
constraint~\cite{bunz2018bulletproofs}) at the cost of larger proofs or slower
verification.
Table~\ref{tab:disclosure} gives the disclosure each constraint class
forces.

\begin{table}[t]
\centering
\small
\setlength{\tabcolsep}{4pt}
\begin{tabular}{@{}p{1.55cm}p{2.5cm}p{1.5cm}@{}}
\toprule
\textbf{Conflict} & \textbf{Min.\ disclosure} & \textbf{Privacy} \\
\midrule
Spatial & location fields & high \\
Temporal & timestamps & high \\
Quantity & amounts (or ZK aggregate) & medium \\
Quality & condition / sensor readings & low \\
Regulatory & certificate / credential (or ZK validity proof) & medium \\
\bottomrule
\end{tabular}
\caption{Disclosure forced by each contradiction class. Some classes
(quality degradation over a chain) resist private checking; we state this
tradeoff rather than claiming uniform privacy.}
\label{tab:disclosure}
\end{table}

\section{Security Analysis}
\label{sec:security}

We state three properties. Each is a claim about the model under the
assumptions of \S\ref{sec:model}; proofs are sketches, and we flag where a
property is only probabilistic.

\begin{theorem}[Detectability]
\label{thm:detect}
Let $h$ honest observers each independently detect a given contradiction
with probability at least $p_{\min}$ within horizon $t$. Then the
probability that at least one detects it is
\[
\Pr[\exists\, \text{detector}] \;\geq\; 1 - (1-p_{\min})^{h}.
\]
If detections arrive as independent Poisson processes with rate at least
$\lambda_{\min}$ per observer, then
$\Pr[\text{detected by } t] \geq 1 - e^{-h\lambda_{\min} t}$.
\end{theorem}
\begin{proof}[Proof sketch]
The events ``observer $j$ fails to detect'' are independent with probability
at most $(1-p_{\min})$; their conjunction has probability at most
$(1-p_{\min})^h$, and the complement is the bound. The Poisson form is the
minimum of $h$ independent exponentials with rate $\lambda_{\min}$.
\end{proof}

The independence of the $h$ per-observer events is an assumption, not a given:
observers running identical deterministic checks over identical local views are
perfectly correlated (all detect or none do), which collapses the bound to a
single Bernoulli trial. We obtain independence from \emph{randomized
same-subject candidate-pair sampling} (\S\ref{sec:eval})---each observer checks
an independently drawn fraction of the subject's candidate pairs---so
$p_{\min}$ is the per-observer
coverage of one draw and the $h$ draws are independent by construction. For a
\emph{specific} contradiction, moreover, an observer contributes $p_{\min}>0$
only if its local view holds \emph{both} ECOs of the pair and it samples
independently; an observer missing either side has $p_{\min}=0$ for that pair
regardless of $h$. So $h$ counts the observers that hold both ECOs and draw
independent samples --- not the raw observer count, and not merely the number
of distinct local views.

The bound scales in the number of \emph{honest observers}, not in the number
of adversaries---a deliberate feature: adding watchtowers strengthens
detection regardless of how many parties misbehave (Figure~\ref{fig:detect}).

\begin{figure}[t]
\centering
\includegraphics[width=\linewidth]{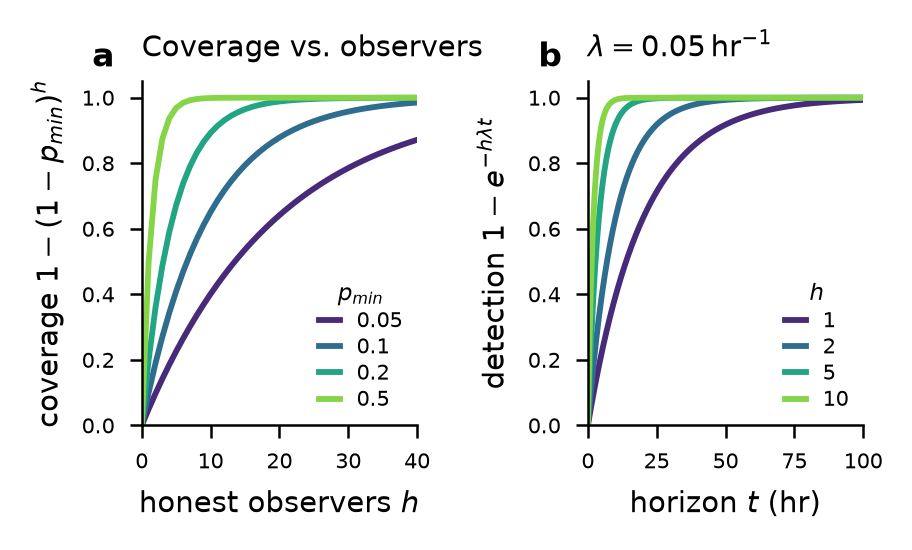}
\caption{Detection model outputs (not measurements). \textbf{Left:}
coverage $1-(1-p_{\min})^{h}$ versus honest observers $h$ for several
per-observer detection probabilities $p_{\min}$. \textbf{Right:} temporal
detection $1-e^{-h\lambda t}$ at $\lambda = 0.05\,\mathrm{hr}^{-1}$ per
observer. Both are consequences of Theorem~\ref{thm:detect} under the stated
parameters, which we do not claim to have measured.}
\label{fig:detect}
\end{figure}

\begin{theorem}[No false accusation]
\label{thm:nofalse}
Under the cryptographic assumptions, an honest party $P$ that (i) never signs
two claims that are jointly unsatisfiable,
and (ii) is never the attributable violator on a self-authored conservation (monotone) chain
(per \S\ref{sec:economics}), cannot be slashed, except with negligible
probability.
\end{theorem}
\begin{proof}[Proof sketch]
By the Blame definition (\S\ref{sec:economics}), slashing fires only on
self-equivocation or monotonic violation. Self-equivocation requires a
verifying CPO whose two ECOs both carry $P$'s valid signature; producing such a
pair without $P$ having signed both requires a signature forgery or a hash
collision, each negligible, so (i) rules it out. Monotonic blame requires $P$
to be attributable as the violator from the pair alone, which (ii) excludes; a
cross-party monotone mismatch in which $P$ authored only one event resolves to
$\bot$. The quality class is the delicate case: a monotone quality
\emph{improvement} is attributable from the pair alone only when the two
condition claims are jointly impossible (self-equivocation on condition, covered
by (i)); an improvement that an unobserved authorized transformation could
explain is not closed over the pair, so a partial-view observer missing $P$'s
authorizing event obtains a $\bot$-blame CPO, not a slash. Signature validity is
not the lever here --- every signature is genuine --- but the Blame rule
(\S\ref{sec:economics}) routes the existential case to $\bot$, so no honest $P$
is slashed by it. Hence no determinate-blame path assigns fault to an honest
$P$.
\end{proof}

\noindent\emph{Remark (key compromise).} The theorem is stated under
existential unforgeability; a \emph{compromised} signing key sits outside it
--- an attacker holding $P$'s key can self-equivocate as $P$ and have $P$
slashed with genuine signatures. Key custody, rotation, and post-compromise
recovery belong to the permissioned identity layer (\S\ref{sec:model}), not to
the protocol's guarantees, and a deployment attaching real financial loss to a
signature should treat them accordingly.

\begin{theorem}[Censorship resistance of proofs]
\label{thm:censor}
Under partial synchrony and Byzantine reliable broadcast with $f < n/3$,
every valid CPO broadcast by an honest node is delivered to all honest nodes
after the global stabilization time and is included in the next checkpoint,
except with negligible probability.
\end{theorem}
\begin{proof}[Proof sketch]
Delivery is the agreement and validity property of Bracha's reliable
broadcast~\cite{bracha1987brb} under $f < n/3$; inclusion follows because
checkpoints commit to all delivered CPOs, and honest nodes constitute a
supermajority.
\end{proof}

This theorem is where the design admits a \emph{bounded} dependence on
agreement, and we scope it precisely: there is no consensus on the event DAG
itself (\S\ref{sec:protocol}) --- emission and detection need none --- but
\emph{finality and pruning} assume a lightweight BFT checkpoint layer among the
$n$ checkpoint participants, i.e.\ periodic agreement on a Merkle root of
delivered CPOs under $f<n/3$. We assume rather than construct that layer; any
standard BFT agreement (e.g.\ PBFT~\cite{castro2002pbft}) instantiates it. A
deployment that forgoes on-chain finality and pruning drops this theorem and
runs with no agreement at all, keeping detection and portable CPO evidence
intact.

\paragraph{Attack surface.}
Table~\ref{tab:attacks} summarizes the qualitative posture. We give
directional detectability---not numeric per-attack rates; the one
quantitative statement is the model of
Theorem~\ref{thm:detect}, which our reference implementation corroborates
against measurement (\S\ref{sec:measured}). The row that matters most is the last: a party
lying consistently emits no contradiction, so the layer's detection is
\emph{none}, and only a complementary oracle can catch it.

\begin{table}[t]
\centering
\small
\setlength{\tabcolsep}{4pt}
\begin{tabular}{@{}p{2.35cm}p{1.7cm}p{2.2cm}@{}}
\toprule
\textbf{Attack} & \textbf{Detected?} & \textbf{Mitigation} \\
\midrule
Overt equivocation & yes (Thm~1) & slash on CPO \\
Collusion ($k$ parties) & shared-gain: yes; suppression: partial (\S\ref{sec:economics}) & scaled stake + watchtowers (\S\ref{sec:economics}) \\
Adaptive evasion & up to sampled coverage $p_{\min}$ (\S\ref{sec:security}) & more watchtowers ($h$) \\
Consistent lying & \textbf{no} & complementary oracle \\
CPO spam (DoS) & n/a & deposit $D$ + rate limit \\
\bottomrule
\end{tabular}
\caption{Qualitative attack posture. ``Detected?'' is directional, not a
measured rate. Consistent lying is out of scope by construction.}
\label{tab:attacks}
\end{table}

\section{Analytical Evaluation}
\label{sec:eval}

This section derives feasibility figures from first principles and states
every assumption; \S\ref{sec:measured} reports measured values from a
single-machine reference implementation alongside the estimates, and finds
them consistent.

\subsection{Storage}
At $10^4$ ECOs/day \emph{entering a participant's local view} (the subject
subgraphs it retains, not only the objects it authors) and an assumed
$\approx 256$~B per object (identifier,
commitment, signature, a bounded reference set, and timestamp), raw growth
is $2.56$~MB/day, or $\approx 0.93$~GB/year; adding CPOs and metadata gives
an order-of-magnitude estimate of $1$~GB per participant per year.
Checkpoint pruning to Merkle roots bounds the retained working set below the
cumulative total. Pruning interacts with detection: a contradiction is
compilable into a CPO only while \emph{both} of its ECOs remain in some honest
observer's unpruned window, so a checkpoint must retain a per-subject digest
sufficient to seed a CPO --- otherwise a contradiction against a pruned-away
claim is no longer actionable, and detection is complete only within the
unpruned horizon. The deterrence condition inherits this: the $p$ in
Eq.~\eqref{eq:deter} is the coverage achieved \emph{within the retention
horizon}, not an asymptotic value, and keeping a contradiction actionable
across a checkpoint requires the digest to retain enough of the opening to
seed a CPO --- a storage cost that partially offsets pruning. The per-object size and reference-count assumptions are
stated so the estimate can be recomputed; the reference implementation
measures $211.9$~B/object over ${\sim}2560$ objects---slightly below the
$256$~B assumption---recomputing it from the actual reference-set distribution
(\S\ref{sec:measured}).

\subsection{Detection coverage}
Coverage follows Theorem~\ref{thm:detect}: it is a function of
$(p_{\min}, h, t)$ only. As an illustration of the tradeoff, lowering the
same-subject candidate-pair sampling fraction reduces per-object work and storage but lowers
$p_{\min}$, moving a deployment down the curves of
Figure~\ref{fig:detect}. We deliberately do not tabulate ``$99.99\%$''
detection figures as if observed; the honest statement is that coverage is
whatever the chosen $(p_{\min}, h, t)$ yield. The reference implementation
measures this directly and the model holds: with a measured
$\hat{p}_{\min} \approx 0.50$, the predicted $1-(1-\hat{p}_{\min})^{h}$ band
overlaps the measured 95\% confidence interval at every observer count
(\S\ref{sec:measured}).

\subsection{Economic viability}
For a transaction value $V_i$ and $\alpha \approx 1.5$, the stake is
$\approx 1.5\,V_i$; its carrying cost is the stake times the cost of
capital. Deterrence holds whenever Eq.~\eqref{eq:deter} is satisfied for the
deployment's $(p, g)$, which for moderate detection ($p \gtrsim 0.5$) needs
$S \gtrsim g$ --- or $S \gtrsim 2g$ under the retained-gain reading of
\S\ref{sec:economics}, the conservative target where fraud is consummated
before detection. The $\alpha \approx 1.5$ illustration therefore assumes the gain
from lying is of the same order as the transaction value ($g \approx V$); where
a single lie captures downstream, market, or reputational value with $g > V$
(\S\ref{sec:economics}), $\alpha$ must be sized from the $g/V$ ratio so that
$S = \alpha V R$ still exceeds $\tfrac{(1-p)}{p}\,g$ --- it is not a fixed
$\approx 1.5$. Small suppliers for whom $S$ is prohibitive can pool stake, at
the cost of shared slashing risk.

\subsection{Measured reference evaluation}
\label{sec:measured}

To check the model rather than merely assert it, we built a single-machine
reference implementation (Python; Ed25519 signatures, prime-order Pedersen
commitments, a zero-knowledge range proof over the homomorphic difference, the
five \textsc{Violates} predicates, the Blame resolver of
\S\ref{sec:economics}, and an $N$-participant / $h$-watchtower simulation with
independent-seed candidate-pair sampling). Every number below is reproducible
from a fixed seed by one command, and the code plus labelled synthetic traces
are released (see \emph{Reproducibility}). It is a \emph{reference}, not a
production build: commitments and the ZK proof run over a 2048-bit MODP integer
group with a naive Sigma protocol rather than the paper's curve-based
Bulletproofs, so the ZK size/time figures are a \textbf{conservative upper
bound}; hashing is BLAKE2b in place of BLAKE3; and all traces are synthetic
with recorded generator parameters.

\paragraph{Detection: the model is confirmed.}
We measure single-observer, single-draw coverage directly as
$\hat{p}_{\min} = 0.501$ (Wilson 95\% CI $[0.465, 0.537]$, $n = 736$).
Predicting $1-(1-\hat{p}_{\min})^{h}$ from that number alone and
\emph{independently} measuring multi-observer coverage, the predicted band
overlaps the measured Wilson CI at every observer count
$h \in \{1,2,3,4,6,8\}$---e.g.\ $h = 3$ measured $0.861\,[0.835, 0.885]$ vs.\
predicted $0.876\,[0.847, 0.901]$, and $h = 8$ measured $1.000\,[0.995, 1.000]$
vs.\ predicted $0.996\,[0.993, 0.998]$. The Detectability theorem's
independent-sampling model is thus empirically corroborated, not merely
assumed; the prediction is not a fit, since $\hat{p}_{\min}$ comes from a
single-draw measurement and the multi-observer points are a separate one.
One caveat keeps the claim honest: the harness \emph{constructs} the
independence the theorem assumes (each watchtower samples with an
independently seeded draw), so the measurement corroborates the composition
law given independence --- it is not field evidence that deployed watchtowers
achieve independent coverage.
Figure~\ref{fig:measured-detect} overlays the measurement on the model band.

\begin{figure}[t]
  \centering
  \includegraphics[width=\linewidth]{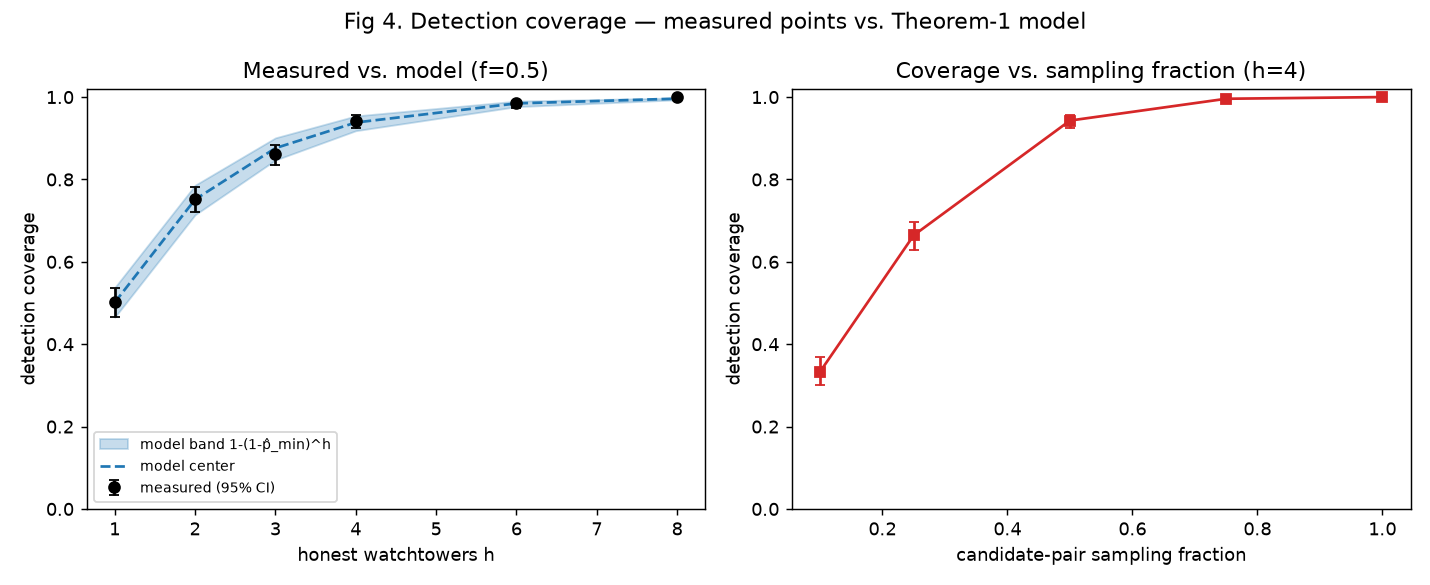}
  \caption{Measured detection coverage vs.\ the Theorem~\ref{thm:detect}
  model (reference implementation, \S\ref{sec:measured}). \emph{Left:} measured
  coverage (black, 95\% CI) against honest watchtowers $h$ overlaid on the model
  band $1-(1-\hat{p}_{\min})^{h}$ from the measured single-observer
  $\hat{p}_{\min}=0.501$; the measurement tracks the band at every $h$.
  \emph{Right:} measured coverage vs.\ the same-subject candidate-pair sampling
  fraction at $h=4$, showing the work/coverage trade-off the sampling knob
  controls. Unlike Figure~\ref{fig:detect}, these are \emph{measurements}, not
  model outputs.}
  \label{fig:measured-detect}
\end{figure}

\paragraph{No false accusation, measured.}
Across the runs, \textbf{zero} honest parties were slashed; precision was
\textbf{1.000} ($280/280$ produced CPOs were true contradictions, no false
positives); and a zero-injection control over 600 honest events produced
\textbf{0} CPOs. Determinate blame broke down as 52 self-equivocation slashes,
7 conservation (quantity) slashes, and 11 $\bot$ resolutions routed to
off-chain adjudication---exercising \emph{both} slashable branches of
\S\ref{sec:economics} and the honest-receiver $\bot$ case, with no honest party
caught. Precision $1.000$ is partly by construction --- a verifying CPO is a
true contradiction absent implementation bugs --- so the informative outcome
here is the Blame routing itself: the 11 cross-party cases went to $\bot$, not
to a slash, and no honest party was ever charged.

\paragraph{Costs.}
Plaintext CPO generation is ${\approx}\,1.5$~\textmu s and verification
${\approx}\,1.1$~ms (median, per constraint class). The watchtower
\textsc{Detect}() cost per incoming ECO grows super-linearly with a subject's
history ($2.7$~\textmu s at history 1 to ${\approx}\,3.0$~ms at history 100),
which is exactly the unbounded baseline the radius limit and candidate-pair
sampling exist to bound. The quantity ZK range proof (conservative MODP/Sigma
bound) is $28.7$~KB with ${\approx}\,2.1$~s prove and verify at a 16-bit range.

\paragraph{Reproducibility.}
The implementation, its tests, a one-command reproduction
(\texttt{scripts/reproduce.py --seed 0}), the raw CSVs, and the generated
results live in the companion repository, archived at
\href{https://doi.org/10.5281/zenodo.21114383}{10.5281/zenodo.21114383}; the
labelled synthetic EPCIS trace set is released as a separate citable benchmark
at \href{https://doi.org/10.5281/zenodo.21114601}{10.5281/zenodo.21114601}.

\section{Illustrative Scenarios}
\label{sec:cases}

Table~\ref{tab:cases} traces six scenarios through the protocol. They are
walked-through illustrations of the mechanism's reach and limits, not
experiments. The final two rows are the honest cases: an ambiguous timing
dispute that yields a CPO but no determinate blame, and a below-threshold
consistent lie that produces no contradiction at all.

\begin{table}[t]
\centering
\small
\setlength{\tabcolsep}{4pt}
\begin{tabular}{@{}p{2.0cm}p{1.0cm}p{1.7cm}p{1.35cm}@{}}
\toprule
\textbf{Scenario} & \textbf{CPO?} & \textbf{Blame} & \textbf{Slash?} \\
\midrule
Falsified shipping record & yes & issuer (equiv.) & yes \\
Cold-chain excursion (transporter's own readings conflict) & yes & transporter (self-equiv.) & yes \\
Counterfeit part (self-inconsistent records) & yes & supplier (self-equiv.) & yes \\
IoT sensor drift & no & --- & no \\
Coordinated delay & yes & ambiguous & if determinate \\
Consistent lie below threshold & no & --- & no \\
\bottomrule
\end{tabular}
\caption{Illustrative scenarios traced through detection, blame, and
slashing. Coordinated silent deception and consistent lying are the
designed-in blind spots.}
\label{tab:cases}
\end{table}

\section{Limitations and Discussion}
\label{sec:limits}

\paragraph{The oracle boundary, restated.}
The protocol detects \emph{provable contradictions}, which is strictly
weaker than verifying truth. A party that is internally consistent and lies
about the physical world---a supplier whose every record agrees with every
other yet all describe a fiction---emits no CPO. This is not an
implementation gap but the definition of the primitive. It is exactly why we
position the layer as complementary to oracle
aggregation~\cite{zhang2020deco,park2023acon2}: those govern ingestion and
can, in principle, catch consistent lies at the point of entry; ours governs
accountability once claims are on the record.

\paragraph{Comparison to oracle staking.}
Chainlink-style oracle networks aggregate independent reports and can slash
node operators for availability or performance faults, but they do not
natively detect \emph{causal contradictions across events} in a
domain-specific log. Our contribution is that downstream detection layer,
not a replacement for aggregation.

\paragraph{Residual limitations.}
Table~\ref{tab:limits} lists the main ones with mitigations. Blame ambiguity
means some real contradictions cannot be adjudicated on-chain; storage grows
without pruning; $p_{\min}$ is a modeling parameter whose real value depends
on watchtower participation; and legal enforceability of a slash sits outside
the protocol.

\begin{table}[t]
\centering
\small
\setlength{\tabcolsep}{4pt}
\begin{tabular}{@{}p{2.05cm}p{2.3cm}p{2.2cm}@{}}
\toprule
\textbf{Limitation} & \textbf{Impact} & \textbf{Mitigation} \\
\midrule
No ground truth & misses consistent lies & pair with oracles \\
Blame ambiguity & some CPOs undecidable & off-chain adjudication \\
Storage growth & $\sim$1\,GB/yr/party & checkpoint pruning \\
Pruning vs.\ detection & contradictions against pruned claims unprovable & retain per-subject digests past checkpoints \\
$p_{\min}$ calibration & affects coverage & watchtower incentives \\
Stake barrier & excludes small firms & stake pooling \\
Legal enforcement & slash $\neq$ judgment & governance overlay \\
\bottomrule
\end{tabular}
\caption{Limitations and mitigations. None is claimed to be eliminated; each
is bounded.}
\label{tab:limits}
\end{table}

\paragraph{Future directions.}
Natural extensions include combining the layer with an authenticated-feed
oracle to close the consistent-lie gap; learning-based prioritization of
which same-subject candidate pairs to check under a fixed sampling budget; cross-domain
interoperability across multiple EPCIS deployments; and a governance model
for adjudicating $\bot$-blame CPOs.

\section{Conclusion}
\label{sec:conclusion}

ECO/CPO-DAG formalizes contradiction detection as a supplemental
accountability layer for adversarial supply chains. Signed claims in a
causally ordered DAG, self-verifying contradiction proofs, determinate-blame
slashing, and commitment-based selective disclosure together let any observer
convert a \emph{provable} inconsistency into an economic penalty, without
consensus and without revealing more than the contradiction requires. We
have been deliberate about scope: the layer does not solve the oracle
problem, it cannot catch a consistent liar, and every quantitative statement
here is a model output rather than a measurement. Within those bounds, it
adds a concrete, cryptographically and economically grounded audit primitive
that complements existing ingestion-time defenses.

\bibliographystyle{plain}
\bibliography{refs}

\end{document}